\documentclass{article}
\usepackage[utf8]{inputenc}
\usepackage[numbers,sort&compress]{natbib}
\usepackage{amsmath,amssymb}

\usepackage{amsthm}

\usepackage{tikz,comment,placeins}
\usetikzlibrary{patterns}
\usepackage{tkz-euclide}
\usepackage{mathtools,graphicx,xcolor}

\usepackage{fmtcount}
\newcounter{bincount}
\newcommand{\binnum}[2][4]{
  \setcounter{bincount}{\numexpr #2\relax}
  \padzeroes[#1]{\binary{bincount}}
}
\makeatletter
\def\mathcolor#1#{\@mathcolor{#1}}
\def\@mathcolor#1#2#3{
  \protect\leavevmode
  \begingroup
    \color#1{#2}#3
  \endgroup
}
\makeatother

\newtheorem{theorem}{Theorem}[section]
\newtheorem{lemma}[theorem]{Lemma}

\newtheorem{proposition}[theorem]{Proposition}
\newtheorem{observation}[theorem]{Observation}
\newtheorem{open}{Open Problem}

\newcommand{\cF}{\mathcal{F}}

\newcommand{\cB}{\mathcal{B}}

\newcommand{\cS}{\mathcal{S}}

\newcommand{\Poly}{\mathbf{P}}

\newcommand{\NP}{\mathbf{NP}}

\title{Role colouring graphs in hereditary classes}
\author{Christopher Purcell \and Puck Rombach\thanks{Department of Mathematics \& Statistics, University of Vermont, USA.}}
\date{}

\begin{document}

\maketitle

\begin{abstract}

    We study the computational complexity of computing role colourings
    of graphs in hereditary classes. We are interested in 
    describing the family of hereditary classes on which a role colouring with $k$ colours can be 
    computed
    in polynomial time.
    In particular, we wish to describe
    the boundary
    between the ``hard'' and ``easy'' classes. The notion of a {\em boundary class} has been
    introduced by Alekseev in order to study such boundaries. Our main results are
    a boundary class for the $k$-role colouring problem and the related
    $k$-coupon colouring problem which has recently received a lot of attention in 
    the literature. The latter result makes use of a technique
    for generating regular graphs of arbitrary girth which may be of independent 
    interest.
\end{abstract}

\section{Introduction}

A {\em role colouring} of a graph $G$ is an assignment $r$ of colours to its vertices such that
if $r(u)=r(v)$ then the image $r(N(u))$ of the neighbourhood of $u$ is identical to the 
image $r(N(v))$ of the neighbourhood of $v$. The concept arises from the study 
of roles in a social network \cite{borgatti1992notions}, and has appeared in the literature
under the names {\em role assignment} \cite{van2010computing,Dourado2016} and {\em regular equivalence} \cite{borgatti1993two,everett1994regular}
and similar names \cite{burt1990detecting,sailer1979structural}.

In a recent paper \cite{PurcellRombach2015}, the present authors proved that the problem of deciding the existence of a role colouring with $k$ colours ($k$-{\sc rcol}) is $\NP$-complete even when the input is restricted to planar graphs. This result is best possible, in the sense that if $P$ is a minor-closed class of graphs, restricting $k$-{\sc rcol} to graphs in $P$ yields a polynomial time solution if and only if at least one planar graph is not in $P$. This is because planar graphs are the unique minimal minor-closed class of graphs of unbounded treewidth, and $k$-role colouring 
can be solved efficiently for bounded treewidth classes.

Our current focus is hereditary classes, where such classifications are not usually possible. Unlike the minor-closed case, hereditary classes are not well founded with respect to the containment relation; there need not be a minimal element in a set of hereditary classes. In order to overcome this difficulty, the concept of a {\em boundary} class was introduced by Alekseev \cite{Alekseev2003} for the maximum independent set problem. For a problem $\Pi$ we denote the family
of graph classes on which $\Pi$ has a polynomial-time solution by $\Pi$-easy.
The importance of the notion of a boundary class is due to the following.
If $X$ is defined by finitely many minimal forbidden induced subgraphs,
then $X$ is outside $\Pi$-easy if and only if $X$ contains a boundary class
for the family $\Pi$-easy (under the assumption that $\Poly \not= \NP$).
One of our two main results is the first boundary class for $k$-{\sc rcol}.

We also consider the special case of a {\em coupon colouring}.
A colouring of a graph $G$ is a coupon colouring if every vertex
has a neighbour of every colour.
Coupon colourings are also known as {\em total domatic partitions}, and have recently received attention from the combinatorial \cite{chen2015coupon,shi2017coupon,ChenJin2017}
and algorithmic \cite{Koivisto2017,lee2016total} communities.
We will denote the problem of deciding the existence of a $k$-coupon colouring 
by $k$-{\sc ccol}.
We give the first boundary class for this problem. 
Our proof makes use of a technique for constructing regular graphs
of arbitrary girth which may be of independent interest.
Explicit constructions for regular graphs of large girth have been known for a long time~\cite{biggs1998constructions,lazebnik1995new}. Our construction is far from extremal, but is meant to increase the girth of graphs while preserving $k$-coupon colourability.

In addition to the main results, we show that every non-trivial $2K_2$-free
graph has a $2$-role colouring and that this colouring can be found
in polynomial time. This is in contrast with the fact that the
problem is known to be $\NP$-hard in this class for $k\geq 4$.

The rest of this paper is organised as follows. In the next section we give some definitions and preliminary discussion including a description of boundary classes. For a more thorough overview we direct the reader to \cite{korpelainen2011boundary}. The boundary classes for $k$-role colouring
and $k$-coupon colouring are given in Section~3 and Section~4 respectively. We conclude the paper
with a discussion of the possibility of other boundary classes and some open problems.

\section{Preliminaries}

In this paper graphs may have loops but not multiple edges. Edges are undirected.
The {\em neighbourhood} of a vertex $N_G(v)$ is the set of vertices that share an edge with $v$ in $G$. This may include $v$ if $\{v\} \in E(G)$. 
The {\em degree} of $v$ is $d(v)=|N_G(v)|$. The minimum degree of all the vertices in $G$ is denoted $\delta(G)$; the maximum is denoted $\Delta(G)$.

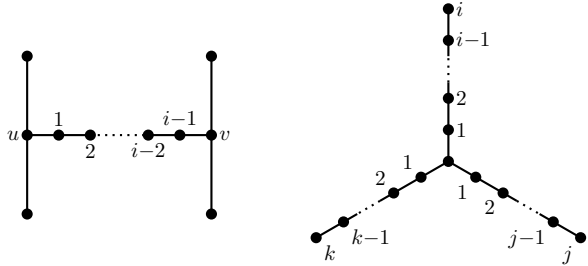
\begin{figure}[!ht]
\centering
\begin{tikzpicture}[thick,scale=0.7, every node/.style={scale=0.8}]
\def\x{1}
\draw   (0,\x) -- (0,3+\x)
        (3.5,\x) -- (3.5,3+\x)
        (0,1.5+\x) -- (1.2,1.5+\x)
        (2.3,1.5+\x) -- (3.5,1.5+\x);
\draw[dotted]   (1.2,1.5+\x)    --  (2.3,1.5+\x);

\fill (0,0+\x) circle (3pt);
\fill (3.5,0+\x) circle (3pt);
\fill (0,3+\x) circle (3pt);
\fill (3.5,3+\x) circle (3pt);
\fill (0,1.5+\x) circle (3pt) node[left] {$u$};
\fill (0.6,1.5+\x) circle (3pt);
\fill (1.2,1.5+\x) circle (3pt);
\fill (2.3,1.5+\x) circle (3pt);
\fill (2.9,1.5+\x) circle (3pt);
\fill (3.5,1.5+\x) circle (3pt) node[right] {$v$};
\draw (.6,1.8+\x) node{$1$};
\draw (1.2,1.2+\x) node{$2$};
\draw (2.3,1.2+\x) node{$i{-}2$};
\draw (2.9,1.8+\x) node{$i{-}1$};

\begin{scope}[style={shift={(8,2)}}]

\draw   (0,0) -- (90:1.5)
        (0,0) -- (210:1.5)
        (0,0) -- (330:1.5)
        (90:2)  --  (90:2.9)
        (210:2) --  (210:2.9)
        (330:2) --  (330:2.9);
        
\draw[dotted]   (90:1.5)    --  (90:2)
                (210:1.5)   --  (210:2)
                (330:1.5)   --  (330:2);
                
\fill   (0,0)       circle  (3pt);
\fill   (90:0.6)    circle  (3pt) node[right] {$1$};
\fill   (90:1.2)    circle  (3pt) node[right] {$2$};
\fill   (90:2.3)    circle  (3pt) node[right] {$i{-}1$};
\fill   (90:2.9)    circle  (3pt) node[right] {$i$};

\fill   (210:0.6)    circle  (3pt) node[above left] {$1$};
\fill   (210:1.2)    circle  (3pt) node[above left] {$2$};
\fill   (210:2.3)    circle  (3pt) node[below right] {$k{-}1$};
\fill   (210:2.9)    circle  (3pt) node[below right] {$k$};

\fill   (330:0.6)    circle  (3pt) node[below left] {$1$};
\fill   (330:1.2)    circle  (3pt) node[below left] {$2$};
\fill   (330:2.3)    circle  (3pt) node[below left] {$j{-}1$};
\fill   (330:2.9)    circle  (3pt) node[below left] {$j$};
            
\end{scope}

\end{tikzpicture}
\caption{Illustration of the graphs $H_i$ and $S_{ijk}$. }\label{fig:HS}
    %\end{center}
\end{figure}

As usual, $P_n,C_n$ and $K_n$ denote the path, cycle and clique on $n$ vertices respectively, and $K_{m,n}$ denotes
the biclique on $m+n$ vertices. 
The graphs $P_n$ and $C_n$ will be particularly important, and it will be useful to define
a consistent ordering of their vertices.
Let $V(P_n)=\{ 1,2,\ldots,n \}$ and let $E(P_n)=\{  \{1,2\}, \{2,3\}, \ldots, \{n-1,n\} \}$.
Let $V(C_n)=V(P_n)$ and let $E(C_n)=E(P_n)\cup\{\{n,1\}\}$. We define two more important
graphs similarly. Let $P_n^*$ be the graph with $V(P_n^*)=V(P_n)$,$E(P_n^*)=E(P_n)\cup\{\{n\}\}$.
Let $P_n^{**}$ be the graph with $V(P_n^*)=V(P_n)$,$E(P_n^*)=E(P_n)\cup\{\{1\},\{n\}\}$.

Let $H$ be the graph on the left of Figure~\ref{fig:HS}, and let $u,v$ be the vertices
of $H$ of degree $3$. We define the graph $H_i$
to be the graph obtained from $H$ by subdividing the edge $\{u,v\}$ $i-1$ times.
Let $S_{ijk}$ be the tree with $3$ leaves at distance $i,j,$ and $k$ respectively 
from the single vertex of degree $3$, for $i,j,k\geq 1$. If one of $i,j$ or $k$ is 
$0$, then let $S_{ijk}=P_{i+j+k+1}$. Figure~\ref{fig:HS} shows $S_{ijk}$ to the right.
An {\em ear} of a graph is a set of vertices of degree 2 that induces a path.

A set of graphs is called a {\em class} if it is closed under isomorphism. 
If a class is also closed under deleting vertices, we say that the class is {\em hereditary}. 
If a hereditary class is closed under deleting edges we say that it is a {\em monotone} class. 
A monotone class which is closed under edge contraction is said to be {\em minor-closed}. 
This paper mainly deals with hereditary classes. It is well known (and easy to verify) that a hereditary class can be characterised by its minimal forbidden induced subgraphs. For a set of graphs $M$, we say that a graph $G$ is $M$-free if no graph in $M$ is an induced subgraph of $G$. The class of $M$-free graphs will be denoted by $Free(M)$. It will be useful to refer to the monotone class of graphs that contain no element of $M$ as a subgraph (not necessarily induced) as $Free_m(M)$.
Let $\cF$ be a family of hereditary classes closed under taking subclasses. 
Suppose $Y_1 \supseteq Y_2 \supseteq \ldots$ is an infinite sequence of graphs
outside $\cF$ and let $Y$ be their intersection. We say that $Y$ is a {\em limit class} for $\cF$,
and a minimal limit class is a {\em boundary class}. A finitely defined graph class $X$ is outside
$\cF$ if and only if it contains none of the boundary classes for $\cF$.

A {\em role colouring} $r:V(G) \to \mathbb{N}^+$ of a graph $G$ is an assignment of colours to its vertices, such that if $r(u)=r(v)$ then $\{ r(u') : u' \in N(u)\} = \{r(v') : v' \in N(v)\}$. In other words, two vertices with the same colour have identical sets of colours in their respective neighbourhoods. 
For a role colouring $r$ we define the {\em role graph} $R$ to be the graph whose
vertex set is $\{i: \exists v\in V(G) r(v)=i\}$ and an edge $\{x,y\}$ if 
a vertex of colour $x$ has a neighbour of colour $y$.
Formally, a role colouring $r$ of a graph $G$ with role
graph $R$ is a {\em locally surjective homomorphism} from $G$ to $R$.
We consider only finite role graphs.
If $R$ has $k$ vertices, we say that $r$ has $k$ colours
and is a $k$-{\em role colouring}.
We may identify the colours of $r$ and the vertices of $R$,
and without loss of generality, these are $\{1,\ldots,k\}$.
In particular, if $R\in\{C_k,P_k,P_k^*,P_k^{**}\}$ then
the colours of $r$ are ordered according to the definition
given above.
We refer to the problem of deciding the existence of a $k$-role colouring for a graph $G$ as $k$-{\sc rcol}, and the problem of deciding the existence of a role-colouring with role graph $R$ for a graph $G$ as $R$-{\sc rcol}.

\begin{observation}\label{obs:connected}
If $G$ is a connected graph with a role colouring $r$, then the role graph $R$ of the colouring is connected.
\end{observation}

\begin{observation}
If $G$ has a non trivial role colouring $r$ with a connected role graph $R$ then $N[v]$ is not monochromatic for any vertex $v$.
\end{observation}

A {\em CNF formula} $\phi$ is a disjunction of clauses $C_1 \wedge \ldots \wedge C_m$,
each of which is a conjunction of literals $C_i = l_1 \vee \ldots \vee l_j$, each of which
is a variable (or its negation) that takes a value $T$ or $F$. A formula
is {\em monotone} if every literal is a variable (and not a negation of a variable).
A {\em not-all-equal satisfying assignment} of $\phi$ is a satisfying assignment of $T,F$ to 
its variables so that each clause has at least one literal whose value is $F$. The problem of deciding whether a (monotone)
formula has a not-all-equal satisfying assignment is denoted by (monotone) {\sc nae-sat}.
This problem is known to be $\NP$-complete \cite{schaefer1978complexity}.

\section{Role Colourings}

Let $S_j$ be the class of $(K_{1,4},C_3,\ldots,C_j,H_1,\ldots,H_j)$-free graphs and let $\cS=\cap_j^{\infty} S_j$. In this section we prove the following theorem.

\begin{theorem}\label{th:boundary}
For all $k\geq 2$, $\cS$ is a boundary class for the $k$-{\sc rcol} problem.
\end{theorem}

\subsection{The limit class}

In order to prove Theorem~\ref{th:boundary}, we first consider the complexity of the $P_k^{**}$-{\sc rcol} problem. We need the following lemma.

\begin{lemma}
\label{lem:period}
Let $G$ be a graph and let $(H,p)$ be in $\{(C_k,k)$,$(P_k,2k-2)$,$(P_k^*,2k-1)$,$(P_k^{**},2k)\}$. Let $\{x,y_1,\ldots,y_{p-1},z\}$ be a path in $G$ such that the degree of $y_i$ is $2$ for all $i$.
Then if $G$ has a valid $H$-role-colouring $r$, we have that $r(x)=r(z)$.
\end{lemma}

\begin{proof}
We give the proof for the case $(H,p)=(P_k^{**},2k)$, the other cases
are analogous. Let $r$ be a $P_k^{**}$-role colouring of $G$.
Consider a walk in $P_k^{**}$ that never uses the same edge twice in a row. It is easy to see that after $2k$ steps, the walk will return to the initial vertex.
Observe that the colours $r(x),r(y_1),\ldots,r(y_{2k-1}),r(z)$ describe such a walk: the two neighbours of $y_i$ must have different colours. Since this walk is $2k$ steps, we have that $r(x)=r(z)$.
\end{proof}

\begin{proposition}\label{prop:pkstarstar}
For all $k\geq 2$, the $P_k^{**}$-{\sc rcol} problem is hard for graphs of vertex degree at most 3.
\end{proposition}

\begin{proof}
We reduce from monotone {\sc nae-sat}. Let $\phi$ be a boolean formula
such that each of its variables appears positively. We will construct a graph $G_\phi$ such that $\phi$ has a not-all-equal satisfying assignment if and only if $G_\phi$ has a $P_k^{**}$-role-colouring. Furthermore, $G_\phi$ will have vertex degree at most 3.

Suppose $\phi$ has $n$ variables $x_1,\ldots,x_n$ and that $x_i$ appears $s_i$ times in $\phi$ (in some arbitrary order). For each variable $x_i$ we include in $G_\phi$ a cycle $\{ x_i^1,\ldots, x_i^{2k(s_i +1)} \}$.
The vertex set of $G_\phi$ will also include the clauses of $\phi$, in some arbitrary order $\{C_1,C_2,\ldots\}$. Suppose $x_i$ appears for the $pth$ time in clause $C_q$. Then we add an edge connecting $x_i^{2pk+2}$ and $C_q$. We will refer to this part of the graph as the {\em formula gadget}.

Additionally, we add a cycle of length $2kn$. We denote this cycle by $\{ v_i^j \}$ with $1 \leq i \leq n$,$1\leq j \leq 2k$. %describe edges
For each $x_i$ we add a path $\{w_i^1,\ldots w_i^{2k-1} \}$ and edges $v_i^1 w_i^1$ and $w_i^{2k-1} x_i^1$. This completes the description of $G_\phi$.

We now argue that $G_\phi$ has a $P_k^{**}$-role-colouring if and only if $\phi$ has a not-all-equal satisfying assignment. Suppose first that $G_\phi$ has such a colouring, $r$. We will say that a set of vertices $X$ is {\em monochromatic} if every vertex in $X$ has the same colour. Observe that by Lemma~\ref{lem:period} the set $\{v_i^1: 1 \leq i \leq n \} \cup \{x_i^1: 1 \leq i \leq n \}$ is monochromatic; let the colour of the vertices in this set be $c$. For each $i$ there are two choices for the colour of $x_i^2$, we call them $c_t$ and $c_f$ (which need not be distinct from $c$). Observe that, for each $i$, we have that $\{x_i^{2pk+2}: 0 \leq p \leq s_i \}$ is monochromatic with colour $c_t$ or $c_f$.

For each clause vertex $C_q$ there are three cases to consider, since these vertices only have neighbours in the sets $\{x_i^{2pk+2}: 0 \leq p \leq s_i \}$. In the first two cases, the neighbourhood of $C_q$ is monochromatic with colour $c_t$ or $c_f$. This is a contradiction since there is no colour that has degree 1 in $P_k^{**}$. In the remaining case, $C_q$ has at least one neighbour of colour $c_t$ and at least one neighbour of colour $c_f$, and $r(C_q)=c$. Now we can obtain a not-all-equal satisfying assignment for $\phi$. Those variables $x_i$ for which  $\{x_i^{2pk+2}: 0 \leq p \leq s_i \}$ has colour $c_b$ can be assigned the truth value $b\in \{t,f\}$. Clearly, each clause has at least one variable assigned $0$ and at least one variable assigned $1$ in this case.

It remains for us to show that a $P_k^{**}$-role-colouring for $G_\phi$ can be constructed from a not-all-equal satisfying assignment for $\phi$. 
Given such an assignment we construct a colouring $r$ as follows. For any vertex $w$ in $\{x_i^1: 1 \leq i \leq n\} \cup \{v_i^1: 1 \leq i \leq n \} \cup \{C_1,C_2,\ldots\}$ we set $r(w)=1$.
Observe that this does not contradict the condition given by Lemma~\ref{lem:period}. 
Since the variable gadget for $x_i$ is a cycle of length a multiple of $2k$, 
we can colour its vertices with colours $1,1,2,3,4,\ldots,k-1,k,k,k-1,\ldots,3,2$ repeated $s_i$ times, 
in one of two possible directions. If $x_i$ is assigned $t$, then we set $r(x_i^2)=1,r(x_i^3)=2$ and so on. 
Otherwise we set $r(x_i^2)=2,r(x_i^3)=3$ and so on. 
Observe that each vertex in the set $\{x_i^{2pk+2}: 0 \leq p \leq s_i \}$ has colour $1$ if $x_i$ is assigned $t$ and with colour $2$ otherwise. By colouring the clause vertices with colour $1$, we have that the colouring is valid for all vertices of the formula gadget.

The cycle in the base gadget is also of length a multiple of $2k$, so we can colour it with colours $1,1,2,3,\ldots,k-1,k,k,k-1,\ldots,3,2$ repeated $n$ times. We restate that $r(v_i^1)=1$ for all $i$. We complete the colouring by giving the vertices $w_i^1,w_i^2,\ldots,w_i^{2k-1}$ the colours $2,3,\ldots,k-1,k,k,k-1,\ldots,3,2,1$ respectively, for all $i$. Since $w_i^1$ has a neighbour $v_i^1$ of colour $1$, and $w_i^{2k-1}$ has a neighbour $x_i^1$ of colour $1$, we see that the colouring is valid for all vertices of the base gadget. This completes the proof. We illustrate this construction in Figure \ref{fig:octopus}, for the boolean formula $\Phi= (x \vee y) \wedge (y \vee z)$, with $k=3$ and role graph $P_3^{**}$.
\end{proof}

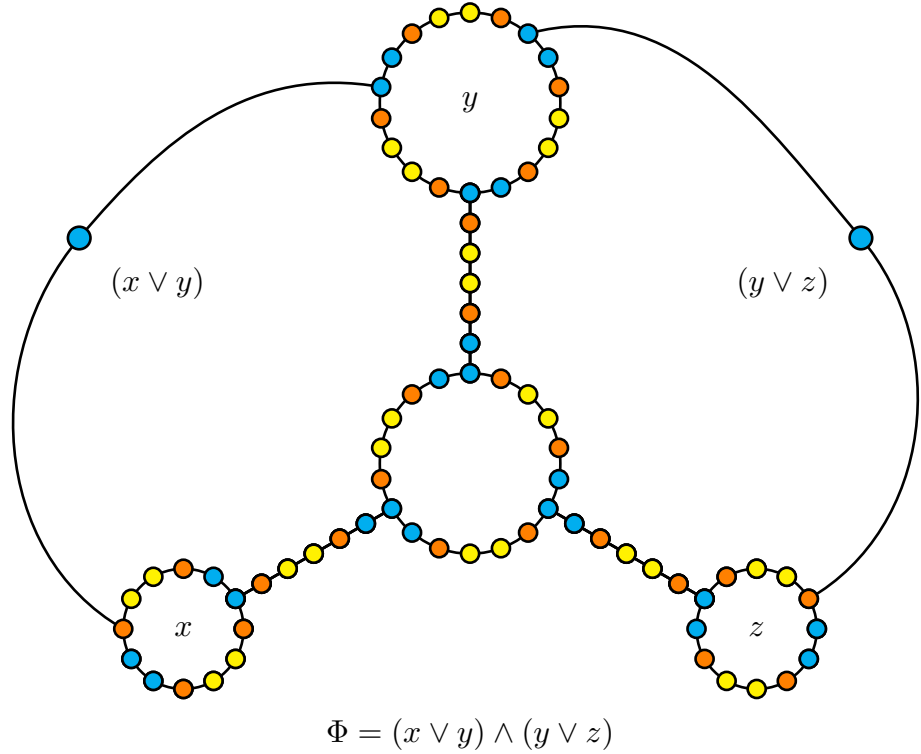
\begin{figure}[!ht]
    \centering
\begin{tikzpicture}[scale=0.6,every node/.style={scale=1.2}]
\draw[black!100,line width=1pt] (0,0) circle (2);
\draw[black!100,line width=1pt] (0,8) circle (2);
	\begin{scope}[style={shift={(-30:6+4/3)}}]
   \draw[black!100,line width=1pt] (0,0) circle (4/3);
   \tkzDefPoint(30:4/3){A}
	\end{scope}
		\begin{scope}[style={shift={(90:8)}}]
   \tkzDefPoint(-90+7*360/18:2){AA}
   \tkzDefPoint(-90+13*360/18:2){BB}
	\end{scope}
		\begin{scope}[style={shift={(210:6+4/3)}}]
   \draw[black!100,line width=1pt] (0,0) circle (4/3);
   \tkzDefPoint(-4/3,0){B}
	\end{scope}
\draw[black!100,line width=1pt](30:10) to[out=-50,in=30] (A);
\draw[black!100,line width=1pt](150:10) to[out=-130,in=150] (B);
\draw[black!100,line width=1pt](30:10) to[out=130,in=15] (AA);
\draw[black!100,line width=1pt](150:10) to[out=50,in=170] (BB);

\foreach \z in {1,2,3}
{
	\draw[fill=cyan,draw=black,line width=1pt] (90+\z*120:2) circle (.2);
	\draw[fill=cyan,draw=black,line width=1pt] (90+\z*120+20:2) circle (.2);
	\draw[fill=orange,draw=black,line width=1pt] (90+\z*120+40:2) circle (.2);
	\draw[fill=yellow,draw=black,line width=1pt] (90+\z*120+60:2) circle (.2);
	\draw[fill=yellow,draw=black,line width=1pt] (90+\z*120+80:2) circle (.2);
	\draw[fill=orange,draw=black,line width=1pt] (90+\z*120+100:2) circle (.2);

	\begin{scope}[style={shift={(90:8)}}]
    \draw[fill=cyan,draw=black,line width=1pt] (-90+\z*120:2) circle (.2);
	\draw[fill=cyan,draw=black,line width=1pt] (-90+\z*120+20:2) circle (.2);
	\draw[fill=orange,draw=black,line width=1pt] (-90+\z*120+40:2) circle (.2);
	\draw[fill=yellow,draw=black,line width=1pt] (-90+\z*120+60:2) circle (.2);
	\draw[fill=yellow,draw=black,line width=1pt] (-90+\z*120+80:2) circle (.2);
	\draw[fill=orange,draw=black,line width=1pt] (-90+\z*120+100:2) circle (.2);
	\end{scope}
	\draw[black!100,line width=1pt] (90:2) -- (90:6);
	\draw[fill=cyan,draw=black,line width=1pt] (90:2+4/6) circle (.2);
	\draw[fill=orange,draw=black,line width=1pt] (90:2+2*4/6) circle (.2);
	\draw[fill=yellow,draw=black,line width=1pt] (90:2+3*4/6) circle (.2);
	\draw[fill=yellow,draw=black,line width=1pt] (90:2+4*4/6) circle (.2);
	\draw[fill=orange,draw=black,line width=1pt] (90:2+5*4/6) circle (.2);
	\draw[fill=cyan,draw=black,line width=1pt] (90:6) circle (.2);

	\draw[black!100,line width=1pt] (210:2) -- (210:6);
	\draw[fill=cyan,draw=black,line width=1pt] (210:2+4/6) circle (.2);
	\draw[fill=orange,draw=black,line width=1pt] (210:2+2*4/6) circle (.2);
	\draw[fill=yellow,draw=black,line width=1pt] (210:2+3*4/6) circle (.2);
	\draw[fill=yellow,draw=black,line width=1pt] (210:2+4*4/6) circle (.2);
	\draw[fill=orange,draw=black,line width=1pt] (210:2+5*4/6) circle (.2);
	\draw[fill=cyan,draw=black,line width=1pt] (210:6) circle (.2);
		\draw[black!100,line width=1pt] (-30:2) -- (-30:6);
	\draw[fill=cyan,draw=black,line width=1pt] (-30:2+4/6) circle (.2);
	\draw[fill=orange,draw=black,line width=1pt] (-30:2+2*4/6) circle (.2);
	\draw[fill=yellow,draw=black,line width=1pt] (-30:2+3*4/6) circle (.2);
	\draw[fill=yellow,draw=black,line width=1pt] (-30:2+4*4/6) circle (.2);
	\draw[fill=orange,draw=black,line width=1pt] (-30:2+5*4/6) circle (.2);
	\draw[fill=cyan,draw=black,line width=1pt] (-30:6) circle (.2);

	\begin{scope}[style={shift={((-30:6+4/3))}}]
    \draw[fill=cyan,draw=black,line width=1pt] (150+\z*180:4/3) circle (.2);
	\draw[fill=cyan,draw=black,line width=1pt] (150+\z*180+30:4/3) circle (.2);
	\draw[fill=orange,draw=black,line width=1pt] (150+\z*180+60:4/3) circle (.2);
	\draw[fill=yellow,draw=black,line width=1pt] (150+\z*180+90:4/3) circle (.2);
	\draw[fill=yellow,draw=black,line width=1pt] (150+\z*180+120:4/3) circle (.2);
	\draw[fill=orange,draw=black,line width=1pt] (150+\z*180+150:4/3) circle (.2);
	\end{scope}

	\begin{scope}[style={shift={((-150:6+4/3))}}]
    \draw[fill=cyan,draw=black,line width=1pt] (30+\z*180:4/3) circle (.2);
	\draw[fill=cyan,draw=black,line width=1pt] (30+\z*180+30:4/3) circle (.2);
	\draw[fill=orange,draw=black,line width=1pt] (30+\z*180+60:4/3) circle (.2);
	\draw[fill=yellow,draw=black,line width=1pt] (30+\z*180+90:4/3) circle (.2);
	\draw[fill=yellow,draw=black,line width=1pt] (30+\z*180+120:4/3) circle (.2);
	\draw[fill=orange,draw=black,line width=1pt] (30+\z*180+150:4/3) circle (.2);
	\end{scope}
}

\draw[fill=cyan,draw=black,line width=1pt] (30:10) circle (.25);
\draw[fill=cyan,draw=black,line width=1pt] (150:10) circle (.25);
\draw[fill=cyan,draw=black,line width=1pt] (90:2) circle (.2);
\draw[fill=cyan,draw=black,line width=1pt] (-30:2) circle (.2);
\draw[fill=cyan,draw=black,line width=1pt] (-150:2) circle (.2);

\draw (0,-6) node{$\Phi= (x \vee y) \wedge (y \vee z)$};
\draw (-150:6+4/3) node{$x$};
\draw (-30:6+4/3) node{$z$};
\draw (90:8) node{$y$};
\draw (150:8) node{$(x \vee y)$};
\draw (30:8) node{$(y \vee z)$};

\end{tikzpicture}
    \caption{Example of $G$ for the boolean formula $\Phi= (x \vee y) \wedge (y \vee z)$, with $k=3$ and role graph $P_3^{**}$. The $P_3^{**}$-role colouring corresponds to the assignment $y=T$ and $x=z=F$.}
    \label{fig:octopus}
\end{figure}

The graph $G_\phi$ may have a $k$-role-colouring with a role graph different to $P_k^{**}$. In particular, there exists a formula $\phi$ with no not-all-equal satisfying assignment that is $k$-role colourable. Thus our construction does not immediately prove that $k$-role-colouring is hard for graphs of vertex degree at most 3. We now show how to overcome this difficulty by modifying $G_\phi$ to fix the role graph. 

Let $H$ be an induced subgraph of $G$ such that exactly one vertex $v$ in $V(H)$ has exactly one neighbour in $V(G)\setminus V(H)$, while all other vertices of $H$ have no neighbours in $V(G)\setminus V(H)$. We say that $H$ is a {\em dangling} induced subgraph of $G$, and that $v$ is the {\em hook} of $H$.

\begin{observation}\label{obs:deg}
Let $r$ be a role colouring of a graph $G$, with role graph $R$. Then the degree of $r(v)$ in $R$ is
at most the degree of $v$ in $G$. Furthermore,
$\delta(G) \leq \delta (R) \leq \Delta(R) \leq \Delta(G)$.
\end{observation}

\begin{lemma}
\label{lem:ear}
Let $G$ have an ear of length at least $2k$ and suppose $G$ is $k$-role-colourable with connected role graph $R$. Then $R\in\{C_k,P_k,P_k^*,P_k^{**}\}$.
\end{lemma}

\begin{proof}

Suppose the contrary. Then there must be a colour $c_1$ in $R$ of degree at least $3$. 
Let $v_1,\ldots,v_{2k}$ be the vertices of an ear. By Observation~\ref{obs:deg},
no vertex of the ear has colour $c_1$. 
For some $j$, assume that no vertex in $v_i$ with $j+1 \leq i \leq 2k-j$
has a colour of distance at most $j$ from $c_1$ in $R$.
If there is no colour of distance $j+1$ from $c_1$ in $R$ then
we have reached a contradiction.
Otherwise, let $c'$ be such a colour, and observe that 
if a vertex $v_i$ with $j+2 \leq i \leq 2k-j-1$ has colour $c'$
then $c'$ has degree 2 in $R$, and one of the two neighbours 
of $v_i$ necessarily has a colour of distance $j$ from $R$,
a contradiction.
By induction, we see that $r(v_{k+1})$ is at distance
strictly more than $k$ from $c_1$, which gives a contradiction
that completes the proof.

\end{proof}

\begin{lemma}
\label{lem:dang}
Let $G$ be a graph with a dangling induced subgraph $H$ isomorphic to $C_{2k-1}$,
and suppose that $G$ is $k$-role-colourable with connected role graph $R$. Then $R\in\{P_k^*,P_k^{**}\}$.
\end{lemma}

\begin{proof}
Let $V(H)=\{v_1,\ldots,v_{p(k)}\}$ and let $v_1$ be the hook of $H$. 
By Lemma~\ref{lem:ear}, $R\in\{C_k,P_k,P_k^*,P_k^{**}\}$. 
Clearly there is no way to colour the vertices of $H$ such
that the role colouring is a $C_{2k}$.
By Lemma~\ref{lem:period},
if $R=P_k$, then $v_1$ and $v_{p(k)}$ have the same colour, which
is a contradiction. Figure \ref{fig:dang} illustrates two copies of $H$ for $k=3$, one that is $P_3^*$-role coloured, and one that is $P_3^{**}$-role coloured.
\end{proof}

\begin{figure}[!ht]
    \centering

\begin{tikzpicture}[scale=.7,every node/.style={scale=1}]

  	\begin{scope}[style={shift={(0,0)}}]
    \draw[fill=black!10,draw=black!20,line width=0pt] (0,0) circle (1.5);
  \draw[black!100,line width=1pt] (0,0) circle (1);
\draw[black!100,line width=1pt,dashed] (0,1) -- (0,2);
\draw[black!100,line width=1pt,dashed] (0,2) -- (1,3);
\draw[black!100,line width=1pt,dashed] (0,2) -- (-1,3);
\draw[fill=cyan,draw=black,line width=1pt] (0,2) circle (.2);
  \draw[fill=cyan,draw=black,line width=1pt] (90+0*72:1) circle (.2);
  \draw[fill=cyan,draw=black,line width=1pt] (90+1*72:1) circle (.2);
  \draw[fill=orange,draw=black,line width=1pt] (90+2*72:1) circle (.2);
  \draw[fill=yellow,draw=black,line width=1pt] (90+3*72:1) circle (.2);
  \draw[fill=orange,draw=black,line width=1pt] (90+4*72:1) circle (.2);
  \draw (0,0) node{$H$};
  \draw (0,-2) node{$R=P_3^*$};
	\end{scope}

  	\begin{scope}[style={shift={(4,0)}}]
    \draw[fill=black!10,draw=black!20,line width=0pt] (0,0) circle (1.5);
  \draw[black!100,line width=1pt] (0,0) circle (1);
\draw[black!100,line width=1pt,dashed] (0,1) -- (0,2);
\draw[black!100,line width=1pt,dashed] (0,2) -- (1,3);
\draw[black!100,line width=1pt,dashed] (0,2) -- (-1,3);
\draw[fill=cyan,draw=black,line width=1pt] (0,2) circle (.2);
  \draw[fill=cyan,draw=black,line width=1pt] (90+0*72:1) circle (.2);
  \draw[fill=orange,draw=black,line width=1pt] (90+1*72:1) circle (.2);
  \draw[fill=yellow,draw=black,line width=1pt] (90+2*72:1) circle (.2);
  \draw[fill=yellow,draw=black,line width=1pt] (90+3*72:1) circle (.2);
  \draw[fill=orange,draw=black,line width=1pt] (90+4*72:1) circle (.2);
  \draw (0,0) node{$H$};
  \draw (0,-2) node{$R=P_3^{**}$};
	\end{scope}
\end{tikzpicture}
    \caption{Examples $H$ for $k=3$, with a $P_k^*$-role colouring and a $P_k^{**}$-role colouring, respectively.}
    \label{fig:dang}
\end{figure}
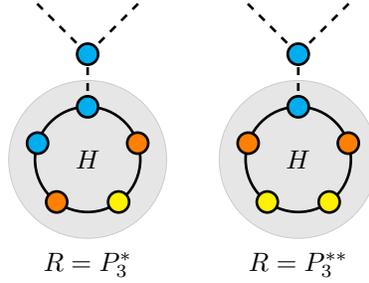

\begin{lemma}
\label{lem:dang2}
Let $p(k)$ as before. Let $G$ be a graph with two dangling induced subgraphs $H,H'$,
each isomorphic to $C_{2k}$, with respective hooks $v_1,v'_1$. Let $u$ be the unique neighbour of $v_1$ in $V(G)\setminus V(H)$ and let $u'$ be the unique neighbour of $v'_1$
in $V(G)\setminus V(H')$. Let $u$ and $u'$ be adjacent in $G$. Suppose $r$ is a $k$-role-colouring of $G$ with connected role graph $R$. Then $R\in\{C_k,P_k,P_k^{**}\}$
\end{lemma}

\begin{proof}
Let $V(H)=\{v_1,\ldots,v_{p(k)}\}$ and $V(H')=\{v'_1,\ldots,v'_{2k}\}$. By Lemma~\ref{lem:ear}, $R \in \{C_k,P_k,P_k^*,P_k^{**} \}$. By Lemma~\ref{lem:period}, if $R=P_k^*$ then $v_1$ and $v_{p(k)}$ have the same colour; i.e., they have colour $k$. Similarly,
$v_2$ has colour $k$. For the same reason, $v'_1,v'_2$ and $v'_{2k}$ have colour $k$.
Since $r$ is a $P_k^*$-colouring, $v_1$ must have a neighbour of colour $k-1$. 
We must have that $r(u)=k-1$ and symmetrically, $r(u')=k-1$, which gives a contradiction. Figure \ref{fig:dang2} illustrates three examples of this construction for $k=3$: one that is $C_3$-role coloured, one that is $P_3^*$-role coloured and one that is $P_3^{**}$-role coloured.
\end{proof}

\begin{figure}[!ht]
    \centering
   \begin{tikzpicture}[scale=.6,every node/.style={scale=1}]
  \draw[fill=black!10,draw=black!20,line width=0pt] (0,0) circle (1.5);
    \draw[fill=black!10,draw=black!20,line width=0pt] (1,4) circle (1.5);
  \draw[black!100,line width=1pt] (0,0) circle (1);
\draw[black!100,line width=1pt,dashed] (0,1) -- (0,2);
\draw[black!100,line width=1pt,dashed] (0,2) -- (1,2);
\draw[black!100,line width=1pt,dashed] (1,2) -- (2,1);
\draw[black!100,line width=1pt,dashed] (0,2) -- (-1,3);
\draw[black!100,line width=1pt,dashed] (1,2) -- (1,3);
\draw[fill=yellow,draw=black,line width=1pt] (0,2) circle (.2);
\draw[fill=orange,draw=black,line width=1pt] (1,2) circle (.2);
  \draw[fill=cyan,draw=black,line width=1pt] (90+0*60:1) circle (.2);
  \draw[fill=orange,draw=black,line width=1pt] (90+1*60:1) circle (.2);
  \draw[fill=yellow,draw=black,line width=1pt] (90+2*60:1) circle (.2);
  \draw[fill=cyan,draw=black,line width=1pt] (90+3*60:1) circle (.2);
  \draw[fill=orange,draw=black,line width=1pt] (90+4*60:1) circle (.2);
    \draw[fill=yellow,draw=black,line width=1pt] (90+5*60:1) circle (.2);
    	\begin{scope}[style={shift={(1,4)}}]
    \draw[black!100,line width=1pt] (0,0) circle (1);
    \draw[fill=cyan,draw=black,line width=1pt] (90+0*60:1) circle (.2);
  \draw[fill=orange,draw=black,line width=1pt] (90+1*60:1) circle (.2);
  \draw[fill=yellow,draw=black,line width=1pt] (90+2*60:1) circle (.2);
  \draw[fill=cyan,draw=black,line width=1pt] (90+3*60:1) circle (.2);
  \draw[fill=orange,draw=black,line width=1pt] (90+4*60:1) circle (.2);
    \draw[fill=yellow,draw=black,line width=1pt] (90+5*60:1) circle (.2);
    	\end{scope}
  \draw (0,0) node{$H$};
   \draw (1,4) node{$H$};
  \draw (0,-2) node{$R=C_3$};

  	\begin{scope}[style={shift={(5,0)}}]
     \draw[fill=black!10,draw=black!20,line width=0pt] (0,0) circle (1.5);
    \draw[fill=black!10,draw=black!20,line width=0pt] (1,4) circle (1.5);
  \draw[black!100,line width=1pt] (0,0) circle (1);
\draw[black!100,line width=1pt,dashed] (0,1) -- (0,2);
\draw[black!100,line width=1pt,dashed] (0,2) -- (1,2);
\draw[black!100,line width=1pt,dashed] (1,2) -- (2,1);
\draw[black!100,line width=1pt,dashed] (0,2) -- (-1,3);
\draw[black!100,line width=1pt,dashed] (1,2) -- (1,3);
\draw[fill=orange,draw=black,line width=1pt] (0,2) circle (.2);
\draw[fill=orange,draw=black,line width=1pt] (1,2) circle (.2);
  \draw[fill=cyan,draw=black,line width=1pt] (90+0*60:1) circle (.2);
  \draw[fill=orange,draw=black,line width=1pt] (90+1*60:1) circle (.2);
  \draw[fill=yellow,draw=black,line width=1pt] (90+2*60:1) circle (.2);
  \draw[fill=orange,draw=black,line width=1pt] (90+3*60:1) circle (.2);
  \draw[fill=cyan,draw=black,line width=1pt] (90+4*60:1) circle (.2);
    \draw[fill=orange,draw=black,line width=1pt] (90+5*60:1) circle (.2);
    	\begin{scope}[style={shift={(1,4)}}]
    \draw[black!100,line width=1pt] (0,0) circle (1);
    \draw[fill=orange,draw=black,line width=1pt] (90+0*60:1) circle (.2);
  \draw[fill=cyan,draw=black,line width=1pt] (90+1*60:1) circle (.2);
  \draw[fill=orange,draw=black,line width=1pt] (90+2*60:1) circle (.2);
  \draw[fill=cyan,draw=black,line width=1pt] (90+3*60:1) circle (.2);
  \draw[fill=orange,draw=black,line width=1pt] (90+4*60:1) circle (.2);
    \draw[fill=yellow,draw=black,line width=1pt] (90+5*60:1) circle (.2);
    	\end{scope}
  \draw (0,0) node{$H$};
   \draw (1,4) node{$H$};
  \draw (0,-2) node{$R=P_3$};

	\end{scope}

  	\begin{scope}[style={shift={(10,0)}}]
    \draw[fill=black!10,draw=black!20,line width=0pt] (0,0) circle (1.5);
    \draw[fill=black!10,draw=black!20,line width=0pt] (1,4) circle (1.5);
  \draw[black!100,line width=1pt] (0,0) circle (1);
\draw[black!100,line width=1pt,dashed] (0,1) -- (0,2);
\draw[black!100,line width=1pt,dashed] (0,2) -- (1,2);
\draw[black!100,line width=1pt,dashed] (1,2) -- (2,1);
\draw[black!100,line width=1pt,dashed] (0,2) -- (-1,3);
\draw[black!100,line width=1pt,dashed] (1,2) -- (1,3);
\draw[fill=cyan,draw=black,line width=1pt] (0,2) circle (.2);
\draw[fill=orange,draw=black,line width=1pt] (1,2) circle (.2);
  \draw[fill=cyan,draw=black,line width=1pt] (90+0*60:1) circle (.2);
  \draw[fill=cyan,draw=black,line width=1pt] (90+1*60:1) circle (.2);
  \draw[fill=orange,draw=black,line width=1pt] (90+2*60:1) circle (.2);
  \draw[fill=yellow,draw=black,line width=1pt] (90+3*60:1) circle (.2);
  \draw[fill=yellow,draw=black,line width=1pt] (90+4*60:1) circle (.2);
    \draw[fill=orange,draw=black,line width=1pt] (90+5*60:1) circle (.2);
    	\begin{scope}[style={shift={(1,4)}}]
    \draw[black!100,line width=1pt] (0,0) circle (1);
    \draw[fill=yellow,draw=black,line width=1pt] (90+0*60:1) circle (.2);
  \draw[fill=yellow,draw=black,line width=1pt] (90+1*60:1) circle (.2);
  \draw[fill=orange,draw=black,line width=1pt] (90+2*60:1) circle (.2);
  \draw[fill=cyan,draw=black,line width=1pt] (90+3*60:1) circle (.2);
  \draw[fill=cyan,draw=black,line width=1pt] (90+4*60:1) circle (.2);
    \draw[fill=orange,draw=black,line width=1pt] (90+5*60:1) circle (.2);
    	\end{scope}
  \draw (0,0) node{$H$};
   \draw (1,4) node{$H$};
  \draw (0,-2) node{$R=P_3^{**}$};

	\end{scope}
\end{tikzpicture}
    \caption{Examples $2H$ for $k=3$, with a $C_3$-role colouring, a $P_3^*$-role colouring and a $P_3^{**}$-role colouring, respectively.}
    \label{fig:dang2}
\end{figure}
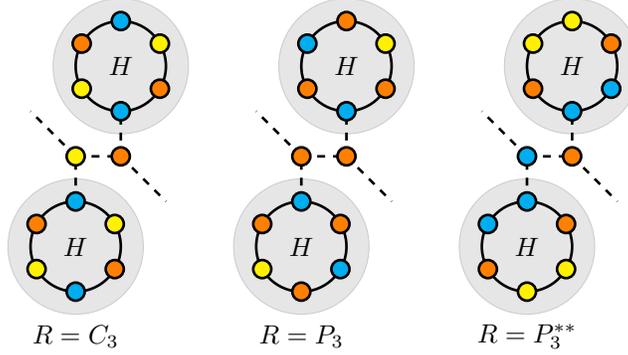

\begin{proposition}\label{prop:krole}
For all $k\geq 2$, the $k$-{\sc rcol} problem is hard for graphs of vertex degree at most 3

\end{proposition}

\begin{proof}
Let $G_\phi$ be the graph constructed in the proof of Proposition~\ref{prop:pkstarstar}. We modify $G_\phi$ to obtain a graph $G'_\phi$ such that $G'_\phi$ has a $k$-role colouring if and only if $\phi$
has a not-all-equal satisfying assignment. To obtain $G'_\phi$ we add to
$G_\phi$ three cycles $C=\{u_1,\ldots,u_{2k-1}\},C'=\{u'_1,\ldots,u'_{2k}\},C''=\{u''_1,\ldots,u''_{2k}\}$.
We also add the edges $u_1 v_n^{2}$, $u'_1 v_n^{3}$ and $u''_1 v_n^{4}$.

Suppose that $G'_\phi$ has a $k$-role colouring $r$. By Lemma~\ref{lem:dang} and Lemma~\ref{lem:dang2}
The role graph $R$ of $r$ is $C_k$ or $P_k^{**}$. Observe that by Lemma~\ref{lem:period} $\{v_i^1: 1\leq i \leq n\} \cup \{ x_i^1 : 1 \leq i \leq n\}$ is monochromatic. The rest of the argument follows as in Proposition~\ref{prop:pkstarstar}: $\phi$ has a not-all-equal satisfying assignment.

Now suppose that $\phi$ has a not-all-equal satisfying assignment. We show that $G'_\phi$ has a $P_k^{**}$-role colouring
and therefore a $k$-role colouring. We know from Proposition~\ref{prop:pkstarstar} that
$G_\phi$ has a $P_k^{**}$-role colouring. We will extend this colouring to $G'_\phi$. We start by colouring $C$.
We set $r(u_1)=1$ (observing that $r(v_n^{2})=1$), and $r(u_2)=r(u_{2k-1})=2$. This forces us to set $r(u_3)=r(u_{2k-2})=3$. If $k=3$ this completes
the colouring of $C$, otherwise we continue the colouring $r(u_i)=r(u_{2k-(i-1)}=i$. Observe that $r(u_k)=k$ and $r(u_{2k-(k-1)})=k$, and $u_k$ is adjacent to $u_{2k-(k-1)}=u_{k+1}$. This completes the colouring of $C$.
To colour $C'$ we set $r(u'_1)=1$ (observing that $r(v_n^{3}=2$). Since $C'$ is a cycle of length $2k$,
we may colour its vertices $1,1,\ldots,k,k,\ldots,3,2$. We colour $C''$ similarly, but starting at $r(u''_1)=2$.
This completes the colouring and the proof.

\end{proof}

We now show how to increase the length of every cycle, and every $H_i$ graph, to 
prove that $\cS$ is a limit class.

\begin{lemma}
\label{lem:subdiv}
Let $G$ be a graph and let $(H,p)$ be in $\{(C_k,k)$,$(P_k,2k-2)$,$(P_k^*,2k-1)$,$(P_k^{**},2k)\}$. Let $e$ be an edge in $G$, and let $G'$ be obtained from $G$ by performing $p$ subdivisions of $e$. Then $G$ is $H$-role-colourable if and only if $G'$ is.
\end{lemma}

\begin{proof}
Clearly if $G$ is $H$-role colourable, then $G'$ is $H$-role colourable.
Let $e=uv$ and let $s_1,\ldots,s_p$ be the vertices added as a result of subdividing $e$, so that $us_1, s_1s_2,\ldots,s_p v$ are all edges of $G'$. Let $r'$ be an $H$-role colouring
of $G'$. By Lemma~\ref{lem:period}, $r'(u)=r'(s_p)$ and $r'(v)=r'(s_1)$. Therefore
we can $H$-role colour $G$ by setting $r(u)=r'(u),r(v)=r'(v)$.
\end{proof}

\begin{proposition}\label{prop:krolelim}
$\cS$ is a limit class for $k$-{\sc rcol}.
\end{proposition}

\begin{proof}

We show that $k$-{\sc rcol} is $\NP$-hard for the class $S_j$ of $(K_{1,4},C_3,\ldots,C_j,$ $H_1,\ldots,H_j)$-free graphs for all $j\geq3$. Fix some arbitrary constant $j$.
We further modify $G'_\phi$ to obtain a graph $G_\phi^j \in S_j$ such that $G_\phi^j$ has a $k$-role colouring
if and only if $\phi$ has a not-all-equal satisfying assignment. Observe that $G'_\phi$ is $(K_{1,4},C_3,H_1,H_2,H_3)$-free.
Let $p(k)=2k(k-1)(2k-1)$, and note that $p(k)$ is a constant. For each edge $e$ of $G'_\phi$, we perform $jp(k)$ subdivisions
of $e$ to obtain $G_\phi^j$. Certainly, $G_\phi^j$ is $(K_{1,4},C_3,\ldots,C_j,H_1,\ldots,H_j)$-free.
The following are equivalent:

\begin{enumerate}
    \item $G'_\phi$ has a $k$-role colouring
    \item $G'_\phi$ has a $P_k^{**}$-role colouring
    \item $G^j_\phi$ has a $P_k^{**}$-role colouring
    \item $G^j_\phi$ has a $k$-role colouring
\end{enumerate}

We have shown that 1.$\Rightarrow$2.; that 3.$\Rightarrow$4. is immediate. By Lemma~\ref{lem:subdiv}, 2.$\Rightarrow$3. and 4.$\Rightarrow$1., as required. Since $G_\phi^j$ has a $k$-role colouring if and only if $G'_\phi$ has a $k$-role colouring
if and only if $\phi$ has a not-all equal satisfying assignment, the proof is complete.

\end{proof}

\subsection{Minimality of the Limit Class}

To prove that the limit class described above is minimal, we introduce the following minimality criterion as a lemma.

\begin{lemma}\label{lem:min}
(\cite{lozin2013boundary})
Let $M$ be the set of minimal forbidden induced subgraphs for a limit class $X$ for some family $\cF$. Then $X$ is minimal if and only if for each $G\in X$ there exists a finite subset $T \subset M$ such that $Free(T\cup\{G\})$ is in $\cF$.
\end{lemma}

Our application of this lemma is identical to that of \cite{lozin2013boundary}; we repeat it here for completeness.
We require the following auxiliary results. 

\begin{theorem}
(\cite{lozin2013boundary}) Let $Y$ be a monotone class of graphs. If there is at least one graph in $\cS \setminus Y$ then the treewidth of $Y$ is bounded.
\end{theorem}

\begin{lemma}
\label{lem:btw}
The $k$-role colouring problem can be solved in polynomial time if the input is restricted to a class of bounded treewidth.
\end{lemma}
\begin{proof}
This is an corollary of the fact that $H$-role colouring can be solved in polynomial time if
$H$ has bounded maximum degree and the input graph $G$ has bounded treewidth \cite{chaplick2015locally}.
\end{proof}

We are now ready to prove Theorem~\ref{th:boundary}. Let $G$ be a graph in $\cS$. Observe
that $G$ is an induced subgraph of $mS_{jjj}$ for some $m$ and $j$; let $M_j$
be the set $\{K_{1,4},C_3,\ldots,C_{2j+1},H_1,\ldots H_{2j+1}\}$. The class of
$(\{G\} \cup M_j)$-free graphs is a proper
subclass of $(\{mS_{jjj}\}\cup M_j)$-free graphs. By Lemma~\ref{lem:min}, it is enough to prove the following
result.

\begin{lemma}
\label{lem:rolp}
For all $m,j\geq1$, the $k$-{\sc rcol} problem is in $\P$ for $(\{mS_{jjj}\}\cup M_j)$-free graphs.
\end{lemma}

\begin{proof}
It is enough to show that $Free(\{mS_{jjj}\}\cup M_j)$ is a proper subclass of $Free_m(mS_{jjj})$.
Let $G$ be a graph in $Free(\{mS_{jjj}\}\cup M_j)$, and suppose $H$ is a subgraph of $G$ isomorphic
to $mS_{jjj}$. By definition of $G$, $V(H)$ does not induce a copy of $mS_{jjj}$. So there
must be at least one additional edge between two vertices of $H$. If there is an edge between
two vertices of the same connected component of $H$ (i.e. of the same $S_{jjj}$), then 
there is necessarily a chordless cycle of length at most $2k+1$ in $G$, which is a contradiction.
Similarly, a small chordless cycle appears if two of the components of $H$ are connected by
at least two edges in $G$. Finally, if two components of $H$ are connected by a single
edge in $G$, an induced subgraph isomorphic to $H_i$ with $i\leq 2k+1$ appears, which is a contradiction.
This completes the proof of Theorem~\ref{th:boundary}.
\end{proof}

We conclude this section by observing that applying Lemma~\ref{lem:subdiv}
directly to the construction in Proposition~\ref{prop:pkstarstar} shows that $\mathcal{S}$ is a boundary class for $H$-role-colouring in the case where $H$ is a $P_k^{**}$ (for $k\geq2$) or a $C_k$ (for $k\geq3$). We note that these graphs are sparse. Previous results on $k$-role-colouring (e.g. \cite{fiala2005complete,PurcellRombach2015}) have also made use of sparse role graphs. In the next section we show that when the role graph is very dense, we can also find a boundary property.

%%%%%%%%%%%%%%%%%%%%%%%%%%%%%%%%%%%%%%%%%%%%%%%%%%%%%%5
% COUPON COLOURING SECTION %%%%%%%%%%%%%%%%%%%%%%%%%%%5
%%%%%%%%%%%%%%%%%%%%%%%%%%%%%%%%%%%%%%%%%%%%%%%%%%%%%%5

\section{Coupon Colourings}

Observe that $2$-{\sc ccol} is $P_2^{**}$-{\sc rcol}, and so $\cS$
is a boundary class in this case. We consider $k$-coupon colourings
for $k\geq 3$.
Let $\mathcal{F}_k$ be the class of forests of maximum vertex degree at most $k$. In this section we prove the following theorem.

\newtheorem*{th:couponboundary}{Theorem \ref{th:couponboundary}}
\begin{theorem}
\label{th:couponboundary}
$\mathcal{F}_k$ is a boundary class for the $k$-{\sc ccol} problem for all $k\geq3$.
\end{theorem}

It is known that $k$-coupon colouring is $\NP$-hard even when the input graph is $k$-regular \cite{Koivisto2017}. This is not a hereditary property of graphs; taking the hereditary closure of the $k$-regular graphs yields the graphs of maximum degree $k$. Let $G$ be a graph of maximum degree $k$. Our goal is to produce a graph $G'$ which is $(K_{1,k+1},C_3,\ldots,C_j)$-free, such that $G$ has a $k$-role colouring if and only if $G'$ has.
The general idea is illustrated in Figure~\ref{fig:implantgen}. We will replace every edge $uv$ with a pair of trees of depth much larger than $j$, such that $u$ is adjacent to $v'$ and $v$ adjacent to $u'$. We add edges between the leaves of these trees to force the colour of $u$ to match the colour of $u'$ (and $v$ to match $v'$) while avoiding cycles of length less than $j$.
We call this operation {\em gemel implantation}. The construction for $k=j=3$ is illustrated in Figure~\ref{fig:implantex}.
Let $G$ be the $3$-regular graph on the left of the figure, and $G'$ be the graph on the right, obtained by a gemel implantation. 
Observe that any $3$-coupon colouring $\phi$ of $G'$ has $\phi(u)=\phi(u')$ (and by symmetry, $\phi(v)=\phi(v')$).
Thus, $G'$ has a 3-coupon colouring if and only if $G'$ has.
Furthermore, there is no triangle in the gemel, nor can there be a triangle of $G'$ including both $u$ and $v$. 
Now let $G''$ be the graph obtained by performing a gemel implantation at every edge of $G$. 
By the previous discussion, $G''$ has a 3-coupon colouring if and only if $G$ has.
Furthermore, $G''$ is $3$-regular and triangle-free, i.e. $(K_{1,4},C_3)$-free.
Clearly $G''$ can be obtained from $G$ in polynomial time. 
Thus, 3-coupon colouring is $\NP$-hard for $(K_{1,4},C_3)$-free graphs.

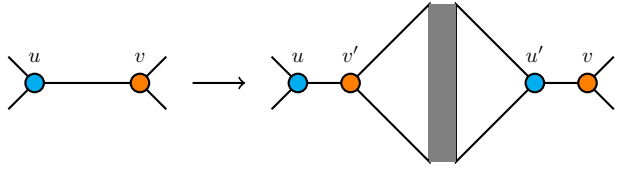
\begin{figure}[t]
    \centering
\begin{tikzpicture}[thick,scale=0.7, every node/.style={scale=0.8}]

\draw   (0,1) -- (2,1)
        (0,1) -- (-0.5,1.5)
        (0,1) -- (-0.5,0.5)
        (2,1) -- (2.5,1.5)
        (2,1) -- (2.5,0.5);

\draw [fill=cyan] (0,1) circle (5pt) node[above=5pt] {$u$};

\draw [fill=orange] (2,1) circle (5pt) node[above=5pt] {$v$};

\draw[->] (3,1) -- (4,1);

\draw (5,1) -- (6,1) -- (7.5,2.5) -- (7.5,-0.5) -- (6,1);

\draw (10.5,1) -- (9.5,1) -- (8,2.5) -- (8,-0.5) -- (9.5,1);

\draw   (5,1) -- (4.5,1.5)
        (5,1) -- (4.5,0.5)
        (10.5,1) -- (11,1.5)
        (10.5,1) -- (11,0.5);

\draw [fill=cyan] (5,1) circle (5pt) node[above=5pt] {$u$};

\draw [fill=cyan] (9.5,1) circle (5pt) node[above=5pt] {$u'$};

\draw [fill=orange] (6,1) circle (5pt) node[above=5pt] {$v'$};

\draw [fill=orange] (10.5,1) circle (5pt) node[above=5pt] {$v$};

\fill [color=gray] (7.5,-0.5) rectangle (8,2.5);

\end{tikzpicture}

    \caption{The gemel implantation.}
    \label{fig:implantgen}
\end{figure}

\begin{figure}[b]
    \centering
 
\begin{tikzpicture}[thick,scale=0.7, every node/.style={scale=0.8}]

\draw   (0,1) -- (2,1)
        (0,1) -- (-0.5,1.5)
        (0,1) -- (-0.5,0.5)
        (2,1) -- (2.5,1.5)
        (2,1) -- (2.5,0.5);

\draw [fill=cyan] (0,1) circle (5pt) node[above=5pt] {$u$};

\fill (2,1) circle (3pt) node[above=5pt] {$v$};

\draw[->] (3,1) -- (4,1);

\draw (5,1) -- (6,1) -- (7,2) -- (8.5,0) -- (9.5,1);

\draw (10.5,1) -- (9.5,1) -- (8.5,2) -- (7,0) -- (6,1);

\draw   (7,2) -- (8.5,2)
        (7,0) -- (8.5,0);

\draw   (5,1) -- (4.5,1.5)
        (5,1) -- (4.5,0.5)
        (10.5,1) -- (11,1.5)
        (10.5,1) -- (11,0.5);

\draw [fill=cyan] (5,1) circle (5pt) node[above=5pt] {$u$};

\draw [fill=cyan] (9.5,1) circle (5pt) node[above=5pt] {$u'$};

\draw [fill=orange] (7,2) circle (5pt);

\draw [fill=yellow] (7,0) circle (5pt);

\fill (6,1) circle (3pt) node[above=5pt] {$v'$};

\fill (10.5,1) circle (3pt) node[above=5pt] {$v$};

\fill (8.5,0) circle (3pt);

\fill (8.5,2) circle (3pt);

\end{tikzpicture}
    \caption{A gemel implantation resulting in a graph of girth at least $4$.}
    \label{fig:implantex}
\end{figure}
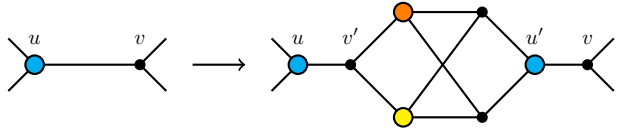

In order to generalise this argument to prove Theorem~\ref{th:couponboundary}, we need to show that we can perform a gemel implantation in $k$-regular graphs that avoids arbitrarily large cycles (and preserves $k$-coupon colourability).
We note that, for the argument to go through, the gemel implantation must be performed in polynomial time. Our construction is of constant size, and therefore gemel implantation at every edge takes $O(n^2)$ time as required. Unfortunately, the constants involved are large; in order to achieve girth $j$, $N$ vertices must be added to each edge where $\log ^* N > j$. 
We suspect there is a more efficient construction, but note that it will
have size at least exponential in $j$.

We demonstrate our construction by first presenting a family of
regular graphs of high girth. These graphs
consist of two cycles and a matching between them. 
We then show how to use the matching to create the gemel implantation.
\subsection{High girth graph $G_{(n,k)}$}

We define a family of graphs with girth $\Omega (\log ^* |V|)$. For ease of notation, let $a_n=a_{(n,r})$. Let $a_0=1$, and let $a_{n+1}=a_{n}+(k-1)^{a_{n}}$. We create a graph $G_{(n,k)}$ as follows. Let $V(G_{(n,k)})=\{0,1\}\times \{ 0,1,\ldots, k-2\} ^{a_{n}}$. We read strings from right to left. 
Let $V_0$ be all the vertices whose strings end with 0, and $V_1$ be all the vertices whose strings end with 1. Add an edge between two vertices in $V_0$ if the number expressed (in base $(k-1)$) by their strings differ by 1 $\pmod{(k-1)^{a_{n}}}$, and do the same for $V_1$. This creates two cycles of length $(k-1)^{a_{n}}$. We add a matching between these cycles to obtain a 3-regular graph. For a vertex $v$ in $V_0$, for $0\leq i \leq n-1$, let $x_i$ be the number expressed by the first $a_i$ digits of $v$. Let $w$ be the vertex in $V_1$ that differ from $v$ only in the digits at positions $a_0+(x_0+1),\ldots a_{n-1}+(x_{n-1}+1)$, namely by adding 1 to these digits $\pmod{k-1}$. Add edges $vw$ for all $v \in V_0$. Let $e:V_0\to V_1$ be the bijective function with $e(v)=w$ as described. Change the last 0 to a 1. Now, $G_{(n,k)}$ is a 3-regular graph.
For example, the graph $G_{(2,3)}$, shown in Figure~\ref{fig:G23} has vertex set $\{ 0,1\} ^4$ and girth 4.  

\begin{lemma}
\label{lem:girth}
$G_{(n,k)}$ has girth $\Omega (\log ^* a_{n})$.
\end{lemma}

\begin{proof}
If a cycle in $G_{(n,k)}$ is entirely in $V_0$ or $V_1$, then it has length $(k-1)^{a_{n}}$. Let $C$ be a cycle in $G_{(n,k)}$ that contains an edge $vw \in [V_0,V_1]$, where $[A,B]$ is the set of edges $ab$ with one endpoint $a \in A$ and one endpoint $b \in B$. Then $v$ and $w$ differ in the digit at position $a_{n-1}+x_{n-1}$ (among others). Suppose a cycle starts at $v,w,\ldots$. To complete a cycle, the digit at $a_{n-1}+x_{n-1}$ must be returned to its state in $v$ (among other digits). This can be done in two ways: along an edge within $V_0$ or $V_1$, or along an edge in $[V_0,V_1]$. Remember that moving along an edge within $V_0$ or $V_1$ changes the represented number of the string by 1, while moving along an edge in $[V_0,V_1]$ changes a set of $n$ digits across the string.\\ 
Let two vertices be \emph{close} if they are distance $\leq (k-1)^{a_{n-2}}$ apart within $G[V_0]$ or in $G[V_1]$, and \emph{far} otherwise.\\
We define a \emph{good} edge in $G_{(n,k)}$ as an edge $\in [V_0,V_1]$ which has one of its endpoints close to another vertex that differs in the digit at position $a_{n-1}+x_{n-1}$.
Other edges in $\in [V_0,V_1]$ are \emph{bad}. 
Good edges have an endpoint that has only $(k-1)$s or only 0s at positions $a_{n-2}+1,\ldots,a_{n-1}$. Without loss of generality, suppose that $vw$ is a good edge, with $v \in V_0,w\in V_1$, such that $v$ has only 0s at positions $a_{n-2}+1,\ldots,a_{n-1}$. Then $w$ has exactly one 1 in the substring at positions $a_{k-2}+1,\ldots,a_{k-1}$ (namely, at position $a_{k-2}+x_{k-2}$, and 0s otherwise. \\
Then, $w$ is far away in $G[V_1]$ from any other endpoint of a good edge. Therefore, there is no cycle of length $\leq (k-1)^{a_{n-2}}$ that uses only good edges from $[V_0,V_1]$.\\
Let $C=(v,w,u,\ldots)$ be a cycle in $G_{(n,k)}$ that uses at least one bad edge $vw \in [V_0,V_1]$. Then, $v$ and $w$ differ at position $a_{n-1}+x_{n-1}$. The path $w,u\ldots$ along $C$ must reach a vertex $x\in V_0$ that is equal to $v$ in the first $a_{n-1}$ bits, or a vertex $y \in V_1$ that is equal to $w$ in the first $a_{n-1}$ bits, such that the edge $xe(x)$ or $ye^{-1}(y)$ changes the digit at position $a_{n-1}+x_{n-1}$ back to its state at $v$. Then either the path $v,w,\ldots,x$ or the path $w,\ldots,y$ correspond to a cycle in $G_{(n-1,k)}$. To complete $C$, by symmetry, there must then also be a path $x,e(x),\ldots,v$ or a path $y,e^{-1}(y),\ldots,v$, which again correspond to a cycle in $G_{(n-1,k)}$. Therefore, 
$$g(G_{(n,k)})\geq \min \{ 2g(G_{(n-1,k)}),(k-1)^{a_{n-2}}  \}, \;\;\; g(G_{(1,k)}\geq3.$$ Then, $$g(G_{(n,k)})=\Omega(2^{n})=\Omega(2^{\log^* a_{n}})=\Omega(\log^* |V|).$$
\end{proof}

\begin{figure}[!ht]
\centering
\begin{tikzpicture}[scale=.4,every node/.style={scale=.9}]
\foreach \z in {0,...,7}
{
\draw[fill=black!100,draw=black,line width=1pt] (0,\z) circle (.1);
\draw[fill=black!100,draw=black,line width=1pt] (5,\z) circle (.1);
\draw (-1.2,\z) node{\binnum{\z}};
\draw (5.8,\z) node{\binnum{\z+8}};
}
\draw[black!100,line width=1pt] (0,0) -- (5,2);
\draw[black!100,line width=1pt] (0,1) -- (5,5);
\draw[black!100,line width=1pt] (0,2) -- (5,0);
\draw[black!100,line width=1pt] (0,3) -- (5,7);
\draw[black!100,line width=1pt] (0,4) -- (5,6);
\draw[black!100,line width=1pt] (0,5) -- (5,1);
\draw[black!100,line width=1pt] (0,6) -- (5,4);
\draw[black!100,line width=1pt] (0,7) -- (5,3);
\draw[black!100,line width=1pt] (0,0) -- (0,7);
\draw[black!100,line width=1pt] (5,0) -- (5,7);
\draw[black!100,line width=1pt] (-3,0) -- (-3,7);
\draw[black!100,line width=1pt] (8,0) -- (8,7);
\draw[black!100,line width=1pt](0,0) to[out=-90,in=-90] (-3,0);
\draw[black!100,line width=1pt](0,7) to[out=90,in=90] (-3,7);
\draw[black!100,line width=1pt](5,0) to[out=-90,in=-90] (8,0);
\draw[black!100,line width=1pt](5,7) to[out=90,in=90] (8,7);
\end{tikzpicture}
\caption{An example of the graph $G_{(n,k)}$ with $k=3$ and $n=2$.}\label{fig:G23}
    %\end{center}
\end{figure}
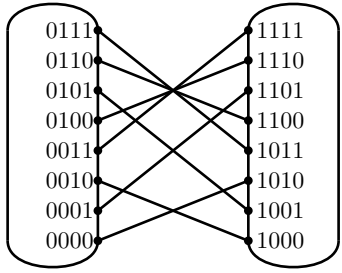

\subsection{Gemel implantation using $G_{(n,k)}$}
Given a $k$-regular graph $G$, we choose a value $n$, and create a $k$-regular graph $G'$ by replacing every edge $vw$ with two edges $uv'$ and $u'v$, and two planar $(k-1)$-ary trees of depth $a_n$, rooted at $v'$ and $u'$. We number the leaves using $k$-ary strings of length $a_n$, by their position in the tree (in the natural way), and add either a 0 or 1 at the end for the tree rooted at $v'$ or $u'$, respectively. We call these leaf sets $V_0$ and $V_1$. We add the matching in $[V_0,V_1]$ based on the girth graph $G_{(n,k)}$. Note that we do not add the cycles within $V_0,V_1$. Finally, we contract all leaf-edges of both trees, and end up with a $k$-regular graph $G'$. (Once we have done this for all edges of $G$.) Figure~\ref{fig:G23ex} shows an example sketch based on $G_{(2,3)}$.

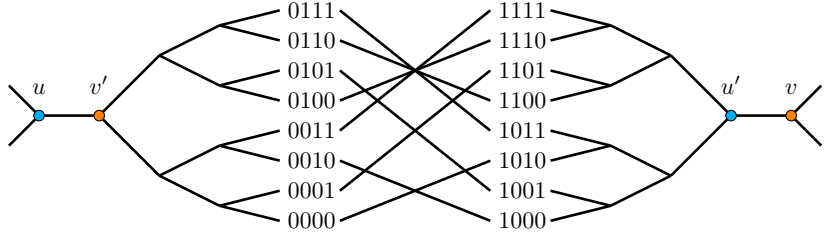
\begin{figure}[!ht]
\centering
\begin{tikzpicture}[scale=.4,every node/.style={scale=.9}]
\foreach \z in {0,...,7}
{
\draw (-1.2,\z) node{\binnum{\z}};
\draw (5.8,\z) node{\binnum{\z+8}};
}
\draw[black!100,line width=1pt] (0,0) -- (5,2);
\draw[black!100,line width=1pt] (0,1) -- (5,5);
\draw[black!100,line width=1pt] (0,2) -- (5,0);
\draw[black!100,line width=1pt] (0,3) -- (5,7);
\draw[black!100,line width=1pt] (0,4) -- (5,6);
\draw[black!100,line width=1pt] (0,5) -- (5,1);
\draw[black!100,line width=1pt] (0,6) -- (5,4);
\draw[black!100,line width=1pt] (0,7) -- (5,3);
\draw[black!100,line width=1pt] (-2,7) -- (-4,6.5);
\draw[black!100,line width=1pt] (-2,6) -- (-4,6.5);
\draw[black!100,line width=1pt] (-2,5) -- (-4,4.5);
\draw[black!100,line width=1pt] (-2,4) -- (-4,4.5);
\draw[black!100,line width=1pt] (-2,3) -- (-4,2.5);
\draw[black!100,line width=1pt] (-2,2) -- (-4,2.5);
\draw[black!100,line width=1pt] (-2,1) -- (-4,0.5);
\draw[black!100,line width=1pt] (-2,0) -- (-4,0.5);
\draw[black!100,line width=1pt] (-4,6.5) -- (-6,5.5);
\draw[black!100,line width=1pt] (-4,4.5) -- (-6,5.5);
\draw[black!100,line width=1pt] (-4,2.5) -- (-6,1.5);
\draw[black!100,line width=1pt] (-4,0.5) -- (-6,1.5);
\draw[black!100,line width=1pt] (-6,5.5) -- (-8,3.5);
\draw[black!100,line width=1pt] (-6,1.5) -- (-8,3.5);
\draw[black!100,line width=1pt] (7,7) -- (9,6.5);
\draw[black!100,line width=1pt] (7,6) -- (9,6.5);
\draw[black!100,line width=1pt] (7,5) -- (9,4.5);
\draw[black!100,line width=1pt] (7,4) -- (9,4.5);
\draw[black!100,line width=1pt] (7,3) -- (9,2.5);
\draw[black!100,line width=1pt] (7,2) -- (9,2.5);
\draw[black!100,line width=1pt] (7,1) -- (9,0.5);
\draw[black!100,line width=1pt] (7,0) -- (9,0.5);
\draw[black!100,line width=1pt] (9,6.5) -- (11,5.5);
\draw[black!100,line width=1pt] (9,4.5) -- (11,5.5);
\draw[black!100,line width=1pt] (9,2.5) -- (11,1.5);
\draw[black!100,line width=1pt] (9,0.5) -- (11,1.5);
\draw[black!100,line width=1pt] (11,5.5) -- (13,3.5);
\draw[black!100,line width=1pt] (11,1.5) -- (13,3.5);
\draw[black!100,line width=1pt] (-10,3.5) -- (-8,3.5);
\draw[black!100,line width=1pt] (-10,3.5) -- (-11,4.5);
\draw[black!100,line width=1pt] (-10,3.5) -- (-11,2.5);
\draw [fill=cyan] (-10,3.5) circle (5pt) node[above=5pt] {$u$};
\draw [fill=orange] (-8,3.5) circle (5pt) node[above=5pt] {$v'$};
\draw[black!100,line width=1pt] (15,3.5) -- (13,3.5);
\draw[black!100,line width=1pt] (15,3.5) -- (16,4.5);
\draw[black!100,line width=1pt] (15,3.5) -- (16,2.5);
\draw [fill=cyan] (13,3.5) circle (5pt) node[above=5pt] {$u'$};
\draw [fill=orange] (15,3.5) circle (5pt) node[above=5pt] {$v$};
\end{tikzpicture}
\caption{An example sketch of a gemel implantation based on $G_{(2,3)}$.}\label{fig:G23ex}
    %\end{center}
\end{figure}

\newtheorem*{prop:ccol}{Proposition \ref{prop:ccol}}
\begin{proposition}\label{prop:ccol}
The $k$-{\sc ccol} problem is $\NP$-complete for graphs of girth at least $g$ for all $g\geq 3$.
\end{proposition}

\begin{proof}
To prove Proposition~\ref{prop:ccol}, we use a construction for the $k$-coupon colouring problem that is a slight variation on the matchings in the graph $G_{(n,k)}$ that forces the pairs ($v,v'$) and ($u,u'$) to be mono-coloured in a $k$-coupon colouring, while still having girth that grows with $n$. This results in a graph $G'$ of high girth that has a $k$-coupon colouring if and only if $G$ has a $k$-coupon colouring. 

Let $a_0=1$, and let $a_{n+1}=a_{n}+2*(k-1)^{a_{n}}$. Note the factor of 2. Given a $k$-regular graph $G$, we choose a value $n$ and, for each edge of $G$, add the two $(k-1)$-ary trees of depth $a_n$, $T_0$ rooted at $v'$ and $T_1$ rooted at $u'$. Label the vertices of the trees by $(k-1)$-ary strings in the natural way, adding a 0 or a 1 at the end for $T_0$ and $T_1$, respectively. The bijective function $e$ that adds the matching now takes two steps. The first step is very similar to $G_{(n,k)}$, except that we change only digits at even positions. For every digit at positions $a_0+2*(x_0+1)-1,\ldots a_{n-1}+2*(x_{n-1}+1)-1$, increase it by $1$ and take $\pmod{k-1}$. Change the last 0 of the string to a 1. Then, remove the second-to last digit (adjacent to the final 1) and insert it in the first position. The first digit becomes the second digit, etc\ldots For example, if $n=1$ and $k=4$, then the string $\mathbf{\mathcolor{cyan}{0}}\mathbf{\mathcolor{orange}{1}}1\mathbf{\mathcolor{orange}{\underline{\overline{1}}}}10002$ has $a_0=1$ and $x_0=2$. The trit at position $a_0+2*(x_0+1)-1=6$ changes to $0+1 \pmod{3}=1$. The last 0 changes to a 1. Finally, the first trit is removed and inserted into second to last position. The resulting string is $\mathbf{\mathcolor{cyan}{1}}1\mathbf{\mathcolor{orange}{\underline{\overline{2}}}}10002\mathbf{\mathcolor{orange}{1}}$.
Without loss of generality, suppose that $v'$ has colour 0 in a $k$-coupon colouring. The colours of $u$, as well as the children of $v'$, are undetermined, but we know that the grand children of $v'$ in $T_0$ must have colors $1,\ldots,k-1$. Then, the colours of the great-grandchildren of $v'$ are undetermined, but we know the sets of colours for each set of great-great-grandchildren of $v'$ that share a common parent. In other words, up to isomorphism, we know the colours of vertices in the tree that are at even distance from $v'$. We describe a function $f_0$ that maps the even-distance vertices of this tree to colours in a $k$-coupon colouring, such that any valid $k$-coupon colouring matches this colouring up to ismorphism. Note that for any vertex $w$ in the tree, $d(v',w)$ is one less than the length of the string. Let 
$$f_0(w)=\sum_{j=0}^{\frac{d(v',w)}{2}-1} (t_{2j+1}+1) \pmod{k},$$
where $t_j$ is the value of the digit at position $j$. For example, the colour of the vertex with string $011110002$ is
$$f_0(011110002)=(2+1)+(0+1)+(1+1)+(1+1) =0 \pmod{4}.$$
To check that this corresponds to a valid $k$-coupon colouring, we need to verify two things:
\begin{itemize}
\item[(1)] Vertices receive colours that are different from their grandparents,
\item[(2)] vertices receive colours that are different from their siblings.
\end{itemize}
Note that the colour of a vertex can be computed from the colour of its grandparent by adding $1+t_1$. Since $0 \leq t_1 \leq k-2$, and $t_1$ distinguishes siblings, these two conditions are met.\\
The leaves of $T_0$ do receive a colour, but they will disappear in the final graph. This colour must be passed on to the vertices at the level above the leaves of $T_1$. For every leaf-parent of $T_1$, this means that all leaves of $T_0$ matched to its children must have the same colour. This is the colour it must adopt to satisfy the leaf-parents in $T_0$, once all leaves have been deleted. Therefore, we need a function $f_1$ that maps vertices of $T_1$ that are at even distance from $u'$ to colours in a $k$-coupon colouring. Let this function be 

$$f_1(w)=\sum_{j=0}^{\frac{d(u',w)-1}{2}} (t_{2j+1}+1) \pmod{k}.$$

For example, the color of the vertex with string $11210002$ is
$$f_1(11210002)=(2+1)+(0+1)+(1+1)+(1+1) =0 \pmod{4}.$$
To check that this corresponds to a valid $k$-coupon colouring, we need to verify three things:
\begin{itemize}
\item[(1)] Vertices receive colours that are different from their grandparents,
\item[(2)] vertices receive colours that are different from their siblings,
\item[(3)] leaf-parents in $T_1$ receive the same colour as the leaves of $T_0$ that are matched to their children.
\end{itemize}
Condition (1) and (2) are verified as before, and condition (3) is verified by the fact that the function of a leaf-parent in $T_1$ sums exactly the same set of numbers as a leaf of $T_0$ that maps to any of its children. \\
Finally, we also note that the children of $u'$ in $T_1$ have colors in the set $\{0,\ldots,k-2\}+1$, implying that $v'$ must have colour 0. By a similar argument, the vertices $u$ and $u'$ must also have the same colour (independent of the colour of $v$ and $v'$).\\

Finally, we note that $g(G')=\Omega (\log (g(G_{(n,k)})))$, because for two leaves (or leaf-parents) $w_1,w_2$ in one of the trees $T_i$ we have $$d_{T_i}(w_1,w_2)=\Omega\left( \log d_{G_{(n,k)}[V_i]}(w_1,w_2)\right).$$ The removal/insertion of the first digit changes this distance only by a constant factor. \end{proof}

\subsection{Minimality}

We apply the minimality criterion to show that the limit class $\cF_k$ is a boundary class for $k$-{\sc ccol}. Let $G$ be a graph in $\cF_k$. Clearly $G$ is an induced subgraph of some perfect $k$-regular tree. 
Let $t$ be the smallest integer such that $G$ is an induced subgraph of $T_k^t$. Since the $M\cup\{G\}$-free graphs form a subset of the $M\cup\{T_k^t\}$-free graphs for any set $M$, it is enough to show that the criterion holds for $T_k^t$ for all  $t$.
We make use of Lemma \ref{lem:deg}.

\begin{lemma}
\label{lem:deg}
Every $(T_k^t,K_{1,k+1},C_3,\ldots,C_{2t+2})$-free graph has a vertex of degree less than $k$.
\end{lemma}

\begin{proof}
Let $H$ be $(T_k^t,K_{1,k+1},C_3,\ldots,C_{2t+2})$-free, and suppose every vertex of $H$ has degree at least $k$. Since $H$ is $(K_{1,k},C_3)$-free, $h$ is in fact $k$-regular. Let $v$ be a vertex of $H$ and consider its radius $t$ neighbourhood $N^t(v)$.%add definition
Since $H$ is $k$-regular and contains no induced cycle of length shorter than $2t+2$, it follows that $N^t(v)$ induces a $T_k^t$. This contradiction completes the proof.
\end{proof}

\begin{proposition}
The $k$-coupon colouring problem can be solved in polynomial time for $(T_k^t,K_{1,k+1},C_3,\ldots,C_{2t+1})$-free graphs.
\end{proposition}

\begin{proof}
By the previous lemma, graphs in this class have no $k$-coupon colouring.
Thus the problem is trivially solvable in polynomial time.
\end{proof}

%%%%%%%%%%%%%%%%%%%%%%%%%%%%%%%%%%%%%%%%%%5
%FINAL SECTION: UNIQUENESS AND DISCUSSION%5
%%%%%%%%%%%%%%%%%%%%%%%%%%%%%%%%%%%%%%%%%%5

\section{Discussion}

We end the paper with a discussion of the uniqueness of our boundary classes and give some open problems. A complete list of boundary classes for a given problem $\Pi$ would characterise the finitely defined hereditary classes for which $\Pi$ is easy (assuming $\Poly \not = \NP$). A set of boundary classes $\{\cB_1,\cB_2,\ldots\}$ for $\Pi$ contains every boundary class if and only if for all $i$ and for all $G\in \cB_i$, $\Pi$ is solvable in polynomial time for $G$-free graphs. 

\subsection{Role colouring $2K_2$-free graphs}

We can immediately see that, for $k \geq 4$, there must be a boundary class for $k$-{\sc rcol}
different from $\cS$, due to the following result.

\begin{theorem}
(\cite{Dourado2016}) For $k\geq 4$, $k$-{\sc rcol} is $\NP$-hard for $(2K_2,C_4,C_5)$-free graphs.
\end{theorem}

We resolve the case $k=2$, leaving $k=3$ open. We need the 
following lemma.

\begin{lemma}\label{lem:2k2struc}
Let $G$ be a connected $2K_2$-free graph. If $G$ has a maximal independent set $I$ of size at least $2$, then the graph induced by $V(G) \setminus I$ is connected.
\end{lemma}

\begin{proof}
Let $u,v\in I$ and consider $S=V\setminus I$. If $|S|=1$ then by definition the graph induced by $S$ is connected. Let $x,y \in S$ and suppose $x$ is a neighbour of $u$. If $y$ is a neighbour of $v$, then $x$ and $y$ are adjacent, otherwise $u,v,x,y$ induce a $2K_2$. If $y$ is not a neighbour of $v$, then there must be a vertex $z \in S$ which is adjacent to $v$, otherwise $G$ is not connected. Now $z$ must be adjacent to $x$, otherwise $u,v,x,z$ induce a $2K_2$. Similarly, $z$ must be adjacent to $y$. Therefore there is a path from $x$ to $y$ as required.
\end{proof}

\begin{theorem}
Every $2K_2$-free graph with at least $2$ vertices has a $2$-role colouring.
Furthermore, such a colouring can be computed in polynomial time.
\end{theorem}

\begin{proof}

If $G$ is disconnected or is a clique, then a $2$-role colouring
is easy to find.
Otherwise, $G$ is a connected graph with a pair
of non-adjacent vertices $u,v$. Let $I$ be a maximal 
independent set containing $u,v$. By Lemma~\ref{lem:2k2struc},
the graph induced by $V(G) \setminus I$ is connected.
So we can colour all vertices of $I$ with colour $1$ and 
all vertices of $V(G) \setminus I$ with colour $2$. 
Since $I$ can be found in polynomial time, this completes the proof.

\end{proof}

\subsection{Coupon colouring claw-free graphs}

Our boundary class for $k$-{\sc ccol} is also not unique due to the following
result.

\begin{theorem}
(\cite{Koivisto2017,heggernes1998partitioning}) The $k$-{\sc ccol} problem is $\NP$-complete for 
$k$-regular claw-free graphs.
\end{theorem}

The authors of \cite{Koivisto2017} only state the theorem for $k$-regular graphs, but it is easy
to see that their construction (which follows Theorem~5 in \cite{heggernes1998partitioning}) is claw-free.

\subsection{Open problems}

The fact that $k$-role colouring $2K_2$-free graphs is easy for $k=2$ but hard for $k\geq 4$ suggests the following questions.

\begin{open}
What is the computational complexity of $3$-role colouring $2K_2$-free graphs?
\end{open}

\begin{open}
Which boundary classes for $k$-role colouring exist inside $2K_2$-free graphs (for $k \geq 4$)?
\end{open}

For $k$-role colouring, the important class of claw-free graphs is a minimal open case.
Also, claw-free graphs must contain a boundary class for $k$-coupon colouring. This suggests
claw-free graphs as a topic of further study.

\begin{open}
What is the computational complexity of $k$-role colouring claw-free graphs?
\end{open}

%Similar questions may be asked regarding $P_k^{**}$-role colouring and $C_k$-role colouring for which $\cS$ is a boundary property by our earlier result and respectively by a simple modification to that construction. In general, we would like to understand the boundary classes for $H$-role colouring, but we pose the following concrete open problem as a starting point.

%Clearly, $\cS$ is not a boundary class in the case of $k$-coupon colouring (where $H=K_k^*$ in the previous open problem). However, we can show that $\cF_k$ is not a unique boundary class for this problem, since there must be another boundary class inside claw-free graphs by the following result.

\begin{open}
Which boundary classes for $k$-coupon colouring exist inside claw-free graphs?
\end{open}

\bibliographystyle{plain}
\bibliography{role}

\end{document}